\documentclass[twoside,12pt]{article}
\pdfoutput=1
\usepackage{times}
\usepackage{fullpage}
\usepackage{color,graphicx}
\usepackage{eurosym}
\usepackage{amssymb}
\usepackage{verbatim}
\usepackage{mathrsfs}

\usepackage{footmisc}

\usepackage{amsmath}
\usepackage{varwidth}

\usepackage{times}
\usepackage{bm}
\usepackage{natbib}

\usepackage{algorithmic}
\usepackage[ruled,vlined,noend]{algorithm2e}

\def\v{{\varepsilon}}

\usepackage{amsthm}
\newtheorem{defn}{Definition}
    
\newtheorem{proposition}[defn]{Proposition} 
\newtheorem{theorem}[defn]{Theorem}

\newtheorem*{prf1}{Proof}

\usepackage{verbatim,enumerate}
\usepackage{graphicx,epsfig}
\usepackage{multirow}
\usepackage{mathrsfs}
\usepackage{amsmath,amssymb,latexsym, amsfonts, amscd}

\def\etr{\text{etr}}
\def\hZ{\hat{Z}}

\def\bkappa{\boldsymbol\kappa}

\def\cY{\mathcal{Y}}

\def\tD{\tilde{D}}
\def\bX{\mathbb{X}}
\def\bU{\mathbb{U}}
\def\htheta{\hat{\theta}}
\newcommand{\dif}{\mathrm{d}}

\def\ml{{p}_{\text{ML}}}
\def\seq{{p}_{\text{seq}}}

\begin{document}

\title{Data augmentation for models based on rejection sampling}

\author{Vinayak Rao\thanks{varao@purdue.edu} \thanks{Department of Statistics, Purdue University, USA }, 
Lizhen Lin
\thanks{Department of Statistics and Data Science, University of Texas at Austin, USA }, 
David Dunson
\thanks{Department of Statistical Science, Duke University, USA}
}

\maketitle

\begin{abstract}
We present a data augmentation scheme to perform Markov chain Monte Carlo inference for models where data generation involves a
rejection sampling algorithm. 
Our idea, which seems to be missing in the literature, is 
a simple scheme to instantiate the rejected proposals preceding each data point. 
The resulting joint probability over observed and rejected variables can be much simpler than the marginal distribution over the
observed variables, which often involves intractable integrals. 
We consider three problems, the first being the modeling of flow-cytometry measurements subject to truncation. 
The second is a Bayesian analysis of the matrix Langevin distribution on the Stiefel manifold, and the third,
Bayesian inference for a nonparametric Gaussian process density model.
The latter two are instances of problems where  Markov chain Monte Carlo inference 
is doubly-intractable. 
Our experiments demonstrate superior performance over state-of-the-art sampling algorithms for such problems.



\end{abstract}

\section{Introduction}


Rejection sampling 
allows sampling from a probability density $p(x)$ by 
constructing an upper bound to $p(x)$,
and accepting or rejecting samples from a density proportional to the bounding envelope. 
The envelope is usually much simpler than $p(x)$, with the number of rejections 
determined by how closely it matches the true density. 

In typical applications, the probability density of interest is indexed by a parameter $\theta$, and we write it as $p(x\mid\theta)$. 
A Bayesian analysis places a prior on $\theta$, and given observations from the likelihood $p(x\mid\theta)$, studies the posterior over $\theta$. 
An intractable likelihood, often with a normalization
constant depending on $\theta$, precludes straightforward Markov chain Monte Carlo inference over $\theta$: calculating a Metropolis-Hastings acceptance ratio
involves evaluating the ratio of two such likelihoods, and is itself intractable. This class of problems is called `doubly intractable' \citep{murray2006},
and existing approaches require the ability to draw exact samples
from $p(x\mid\theta)$, or to obtain positive unbiased estimates of $p(x\mid\theta)$.

We describe a different approach that is applicable when $p(x\mid\theta)$ has an associated rejection sampling algorithm.
Our idea is to instantiate the rejected proposals preceding each observation, resulting in an augmented state-space on which we run a Markov chain.
Including the rejected proposals 
can eliminate any intractable terms, and allow the application of standard techniques \citep{adams_gpds}.
Importantly, we show that conditioned on the observations, it is straightforward to independently sample the number and values of the rejected proposals:
this just requires running the rejection sampler to generate as many acceptances as there are observations, with all rejected proposals kept.
The ability to produce a conditionally independent draw of these variables is important when posterior updates of some parameters are intractable,
while others are simple. In such a situation, we introduce the rejected variables only when we need 
to carry out the intractable updates, after which we discard them and carry out the simpler updates.

A particular application of our algorithm is parameter inference for probability distributions truncated
to sets like the positive orthant, the simplex, or the unit sphere.
Such distributions correspond to sampling proposals from the untruncated distribution and rejecting those outside the domain of interest. 
We consider an application from flow cytometry where this representation is the actual data collection process.
Truncated distributions also arise in diverse applications like measured time-to-infection \citep{Goeth09}, where times larger than a year are truncated,
mortality data \citep{Alai2013}, annuity valuation for truncated lifetimes  \citep{Alai2013}, and
stock price changes \citep{aban06}. One approach for such problems was proposed in \cite{leich09}, though their algorithm samples from an
approximation to the posterior distribution of interest. Our algorithm provides a simple and general way to apply the machinery of Bayesian inference 
to such problems.

\section{Rejection sampling}
Consider a probability density $p(x\mid\theta) = {f(x,\theta)}/{Z(\theta)}$ on some space $\mathbb{X}$, with the parameter $\theta$ 
taking values in $\Theta$.
We assume that the normalization constant $Z(\theta)$ is difficult to evaluate, so that na\"{\i}vely sampling from $p(x\mid\theta)$ is
not easy. We also assume there exists a second, simpler density $q(x\mid\theta) \ge  f(x, \theta)/M$ for all $x$, and for some positive $M$.


Rejection sampling generates samples distributed as $p(\cdot\mid\theta)$ by first proposing samples from $q(\cdot\mid\theta)$. A draw $y$ from 
$q(\cdot\mid\theta)$ is accepted with probability ${ f(y, \theta)}/\left\{M q(y\mid\theta)\right\}$. Let there be $r$ rejected proposals preceeding an accepted sample 
$x$, and denote
them by $\mathcal{Y} = \{y_1, \cdots, y_r \}$ where $r$ itself is a random variable. Write $|\cY| = r$, so that the joint probability is
\begin{align}
  p(\mathcal{Y}, x) & = \left[ \prod_{i=1}^{|\cY|} q(y_i\mid\theta) \left\{1 - \frac{f(y_i, \theta)}{M q(y_i\mid\theta)}\right\} \right]
                        q(x\mid\theta) \left\{ \frac{ f(x, \theta)}{M q(x\mid\theta)} \right\} \nonumber \\
                    & =  \frac{f(x, \theta)}{M} \prod_{i=1}^{|\mathcal{Y}|}  \left\{(q(y_i\mid\theta) - \frac{ f(y_i, \theta)}{M} \right\}. \label{eq:rej_jnt}
\end{align}

It is well known that this procedure recovers samples from $p(x\mid\theta)$, so that the expression above has the correct marginal distribution over $x$ \citep{Robert05}.
Later, we will need to sample the rejected variables $\cY$ given an observation $x$ drawn from $p(\cdot\mid\theta)$. 
Simulating from $p(\cY\mid x,\theta)$ involves the two steps in Algorithm \ref{alg:rej_sim}.
Algorithm~\ref{alg:rej_sim} relies on Proposition~\ref{prop:rej_post} about $p(\mathcal{Y}\mid x,\theta)$ which we prove in
the appendix. 
{
\begin{algorithm}
\caption{Algorithm to sample from $p(\cY\mid x,\theta)$ } \label{alg:rej_sim}
\begin{tabular}{p{1.4cm}p{12.2cm}}
\textbf{Input:}  & A sample $x$, and the parameter value ${\theta}$ \\
\textbf{Output:} & The set of rejected proposals $\cY$ preceeding $x$\\
\hline
\end{tabular}
\begin{algorithmic}[1]
  \STATE Draw sample $y_i$ independently from $q(\cdot\mid\theta)$ until a point $\hat{x}$ is accepted.
  \STATE Discard $\hat{x}$, and treat the preceeding rejected proposals  as $\cY$.
\end{algorithmic}
\end{algorithm}
}

\begin{proposition}
  The set of rejected samples $\cY$ preceding an accepted sample $x$ is independent of $x$: $p(\cY\mid\theta,x) = p(\cY\mid\theta)$. 
We can thus assign $x$ the set $\widehat{\cY}$ of another sample, $\hat{x}$.\label{prop:rej_post}
\end{proposition}  

\section{Bayesian Inference}
\subsection{Sampling by introducing rejected proposals of the rejection sampler}  \label{sec:latent_hist}
Given observations $X = \{x_1, \cdots, x_n\}$, and a prior $p(\theta)$, Bayesian inference typically uses Markov chain Monte Carlo to sample from
the intractable posterior $p(\theta\mid X)$. 
Split $\theta$ as $(\theta_1, \theta_2)$ so that
the normalization constant factors as $Z(\theta) = Z_1(\theta_1) Z_2(\theta_2)$,
with the first term $Z_1$ simple to evaluate, and $Z_2$ intractable.
Updating $\theta_1$ with $\theta_2$ fixed is easy, and there are situations where we can place a conjugate
prior on $\theta_1$. 
Inference over $\theta_2$ is a doubly-intractable problem. 

We assume that $p(x\mid\theta)$ has an associated rejection sampling algorithm with proposal density 
$q(x\mid\theta) \ge f(x,\theta)/M$. 
For the $i${th} observation $x_i$, write the preceding set of
rejected samples as $\cY_i = \{y_{i1}, \ldots, y_{i{|\cY_i|}}\}$.
The joint density of all samples, both rejected and accepted, is then
\begin{align}
 P(x_1, \cY_1,\ldots, x_n, \cY_n) &= \prod_{i=1}^n \frac{f(x_i, \theta)}{M} 
\prod_{j=1}^{|\mathcal{Y}_i|}  \left\{q(y_{ij}\mid\theta) - \frac{f(y_{ij}, \theta)}{M}\right\}. \nonumber
\end{align}
This does not involve any intractable terms, so that standard techniques can be applied to update $\theta$. To introduce the
rejected proposals $\cY_i$, we simply follow Algorithm \ref{alg:rej_sim}: draw proposals from $q(\cdot\mid\theta)$ until we have $n$ acceptances, 
with the $i${th} batch of rejected proposals forming the set $\cY_i$.

The ability to produce conditionally independent draws of $\cY$ is important when, 
for instance, there exists a conjugate prior $p_1(\theta_1)$ on $\theta_1$ for the likelihood 
$p(x\mid\theta_1,\theta_2)$. Introducing the rejected proposals $\cY_i$
breaks this conjugacy, and the resulting complications in updating $\theta_1$ can slow down mixing, especially when $\theta_1$ is
high dimensional.
A much cleaner solution is to sample $\theta_1$ from its conditional posterior $p(\theta_1\mid X,\theta_2)$, introducing the
auxiliary variables only when needed to update $\theta_2$. After updating $\theta_2$, they can then be discarded.
Algorithm \ref{alg:rej_post} describes this.

{
\begin{algorithm}
\caption{An iteration of the Markov chain for posterior inference over $\theta = (\theta_1, \theta_2)$} \label{alg:rej_post}
\begin{tabular}{p{1.4cm}p{12.2cm}}
\textbf{Input:}  & The observations $X$, and the current parameter values $({\theta}_1,{\theta}_2)$ \\
\textbf{Output:} & New parameter values $(\tilde{\theta}_1,\tilde{\theta}_2)$ \\
\hline
\end{tabular}
\begin{algorithmic}[1]
  \STATE Run Algorithm \ref{alg:rej_sim} $|X|$ times, keeping all the rejected proposals $\displaystyle {\cY = \cup_{i=1}^{|X|} \cY_i}$.
  \STATE Update $\theta_2$ to $\tilde{\theta}_2$ with a Markov kernel having $p({\theta}_2\mid X,\cY,\theta_1)$ as stationary distribution.
  \STATE Discard the rejected proposals $\cY$.
  \STATE Sample a new value of ${\theta}_1$ from its posterior $p(\theta_1\mid X,\tilde{\theta}_2)$.
\end{algorithmic}
\end{algorithm}
}

\subsection{Related work}
 One of the simplest and most widely applicable Markov chain Monte Carlo algorithms for doubly-intractable distributions is the exchange sampler 
of \cite{murray2006}. Simplifying an earlier idea by \cite{Moller2006}, this algorithm effectively amounts to the following:
 given the current parameter $\theta_{curr}$, propose a new parameter $\theta_{new}$ according to some proposal distribution.
Additionally, generate a dataset of $n$ `pseudo-observations' $\{\hat{x}_i\}$ from $p(x\mid\theta_{new})$.
The exchange algorithm then proposes swapping the parameters associated with datasets. \cite{murray2006} show that 
all intractable terms cancel out in the resulting acceptance 
probability, and that the resulting Markov chain has the correct stationary distribution.

While the exchange algorithm is applicable whenever one can sample from the likelihood $p(x\mid\theta)$, 
it does not exploit the mechanism used to produce these samples. 
When the latter is a rejection sampling algorithm, each pseudo-observation is preceeded by a sequence of rejected proposals.
These are all discarded, and only the accepted proposals are used to evaluate the new parameter $\theta_{new}$.
By contrast our algorithm explicitly instantiates these rejected proposals, so that they can be used to make \emph{good} proposals. 
In our experiments, we use a Hamiltonian Monte Carlo sampler on the augmented space
and exploit gradient information to make non-local moves with high probability of acceptance.
For reasonable acceptance probabilities under the exchange sampler, one must make local updates to $\theta$, or resort to complicated annealing schemes.

Another framework for doubly intractable distributions is the pseudo-marginal approach of \cite{AndRob10}. The idea here is that even if we 
cannot exactly evaluate the acceptance probability, it is sufficient to use a positive, unbiased estimate: this will still result in a Markov chain with the 
correct stationary distribution. In our case, instead of requiring an unbiased estimate, we require an upper bound $M$ to the normalization constant $Z(\theta)$.
Additionally, like the exchange sampler, the pseudo-marginal method provides a mechanism to evaluate a proposed parameter $\theta_{new}$; how to make good 
proposals is less obvious. Other papers are \cite{beskos05} and \cite{walker11}, the latter requiring a bound on the target density of interest.

 Most closely related to our ideas is a sampler from \cite{adams_gpds}, see also  Section~\ref{sec:gpds}. Their problem also
involved inferences on the parameters governing the output of a rejection sampling algorithm. Like us, they  proceeded by augmenting the state space to
include the rejected proposals $\cY$, and like us, given these auxiliary variables, they used Hamiltonian Monte Carlo to efficiently update parameters. 
However, rather than generating independent realizations of $\cY$ when needed, \cite{adams_gpds}
outlined a set of Markov transition operators to perturb the current configuration of $\cY$,
while maintaining the correct
stationary distribution. With prespecified probabilities, they proposed adding a new variable to $\cY$, deleting a variable from $\cY$
and perturbing the value of an existing element in $\cY$. These local updates to $\cY$ can slow down Markov chain mixing, require the user to specify
a number of parameters, and  also involve calculating Metropolis-Hastings acceptance probabilities for each local step. Furthermore, the Markov nature 
of their updates require them to maintain the rejected proposals at all times, complicating inferences over other parameters.
Our algorithm is much simpler and cleaner.



\section{Convergence properties}

Write the Markov transition density of our chain as $k(\htheta\mid\theta)$, and the $m$-fold transition density as $k^m(\htheta\mid\theta)$.
For simplicity, we suppress that these depend on $X$. The Markov chain is uniformly
ergodic if constants $C$ and $\rho$ exist such that for all $m$ and $\theta$,
$
 \int_{\Theta} | p(\htheta\mid X) - k^m(\htheta\mid\theta)| \mathrm{d}\htheta \le C \rho^m. \nonumber
 $
The term to the left is twice the total variation distance between the desired posterior and the state of the Markov chain initialized at $\theta$ after $m$ iterations. 
Small values of $\rho$ imply faster mixing.

The following minorization condition is sufficient for uniform ergodicity \citep{jones2001}: there exists a probability density $h(\htheta)$ and a $\delta > 0$ such that
for all $\theta,\htheta \in \Theta$, 
\begin{align}
 k(\htheta\mid\theta) \ge \delta h(\htheta). \label{eq:unif_erg}
\end{align}
When this holds, the mixing rate $\rho \le 1-\delta$, so that a large $\delta$ implies rapid mixing.

Our Markov transition density first introduces the rejected proposals $\cY$, and then conditionally
updates $\theta$. The set $\cY_i$ preceeding the $i$th observation takes values in the union space 
$\displaystyle \bU \equiv \cup_{r=0}^{\infty} \bX^r$. 
The output of the rejection sampler, including the $i$th observation, lies in the product space $\bU \times \bX$ with density given by equation \eqref{eq:rej_jnt},
so that any $(\cY, x) \in (\bU \times \bX)$ has probability
\begin{align}
 p(\cY,  x\mid\theta) = \frac{f(x, \theta)}{M}  \lambda(\mathrm{d}x) \prod_{i=1}^{|\mathcal{Y}|}  \left\{q(y_i\mid\theta) - \frac{ f(y_i, \theta)}{M} \right\} \lambda(\mathrm{d}y_i). \nonumber
\end{align}
Here, $\lambda$ is the measure with respect to which the densities $f$ and $q$ are defined, and it is easy to see that the above quantity 
integrates to $1$. From Bayes' rule, the conditional density over $\cY$ is 
\begin{align}
 p(\cY \mid x, \theta) = \frac{Z(\theta)}{M}  \prod_{i=1}^{|\mathcal{Y}|}  \left\{q(y_i\mid\theta) - \frac{ f(y_i, \theta)}{M} \right\} \lambda(\mathrm{d}y_i). \nonumber
\end{align}
The fact that the right hand side does not depend on $x$ is an alternate proof of Proposition \ref{prop:rej_post}.
This density characterizes the data augmentation step of our sampling algorithm. In practice, we need as many draws from this density as there are observations. 

The next step involves updating $\theta$ given $(\cY, X,\theta)$, and depends on the problem at hand.
We simplify matters by assuming we can 
sample from $p({\theta}\mid\cY, X)$ independently of the old $\theta$: this is the classical data augmentation algorithm. We also assume that 
the functions $f(\cdot,\theta)$ and $q(\cdot\mid\theta)$ are uniformly bounded from above and below by finite, positive quantities $(B_f, b_f)$ and 
$(B_q, b_q)$ respectively, and that $\int_{\bX} \lambda(\mathrm{d}x) < \infty$.
It follows that there exists positive numbers $r$ and $R$ that minimize $1-\frac{f(x,\theta)}{\{MZ(\theta)\}}$ and $\frac{Z(\theta)}{M}$.
We can now state our result.
\begin{theorem}
  Assume that $\int_{\bX} \lambda(\mathrm{d}x) < \infty$ and that positive bounds $b_f, B_f, b_q, B_q$ exist with $r$ and $R$ as defined
  earlier. 
Further assume we can sample from the conditional 
$p({\theta}\mid\cY, X)$. Then our data augmentation algorithm is uniformly ergodic with mixing rate $\rho$ upper bounded by
$ \rho = 1-\left\{\frac{b_f}{B_f\left( \beta + R^{-1}\right)}\right\}^n$, 
where $\beta = b_qr/B_q$ and $n$ is the number of observations.  \label{thrm:conv_rate}
\end{theorem}  
Despite our assumptions, our theorem has a number of useful implications. The ratio $b_f/B_f$ is a measure of how `flat' the function $f$ is,
and the closer it is to one, the more efficient rejection sampling for $f$ can be. From our result, the smaller the value, the larger the bound on $\rho$,
suggesting slower mixing. This is intuitive: more rejected proposals $\cY$ will increase coupling between successive $\theta$'s in the 
Markov chain.
On the other hand, a small $b_q/B_q$ suggests a proposal distribution tailored to $f$, and our result shows this implies faster mixing.
The numbers $r$ and $1/R$ are measures of mismatch between the the target and proposal density, with small values giving better mixing.
Finally, more observations $n$ result in slower mixing. Again, from the construction of our chain this makes sense. 
We suspect this last property holds for most exact samplers
for doubly-intractable distributions, though we are unaware of any such result.

Even without assuming we can sample from $p(\theta\mid\cY,X)$, our ability to 
sample $\cY$ independently means that the marginal chain over $\theta$ is still Markovian. By contrast, existing approaches (\cite{adams_gpds, walker11}) only produce dependent updates
in the complicated auxiliary space: they sample from $p(\hat{\cY}\mid\theta, \cY,X)$ by making local updates to $\cY$. Consequently, these chains are Markovian only in the 
complicated augmented 
space, and the marginal processes over $\theta$ have long-term dependencies. Besides affecting mixing, this can also complicate analysis.

In the following sections, we apply our sampling algorithm to three problems, one involving a Bayesian analysis of flow-cytometry data, the second,
Bayesian inference for the matrix Langevin distribution, and the last the
Gaussian process density sampler of \cite{adams_gpds}. 

\section{Flow cytometry data} \label{sec:flow}



  \begin{figure}
  \centering
    \includegraphics[width=.27\textwidth]{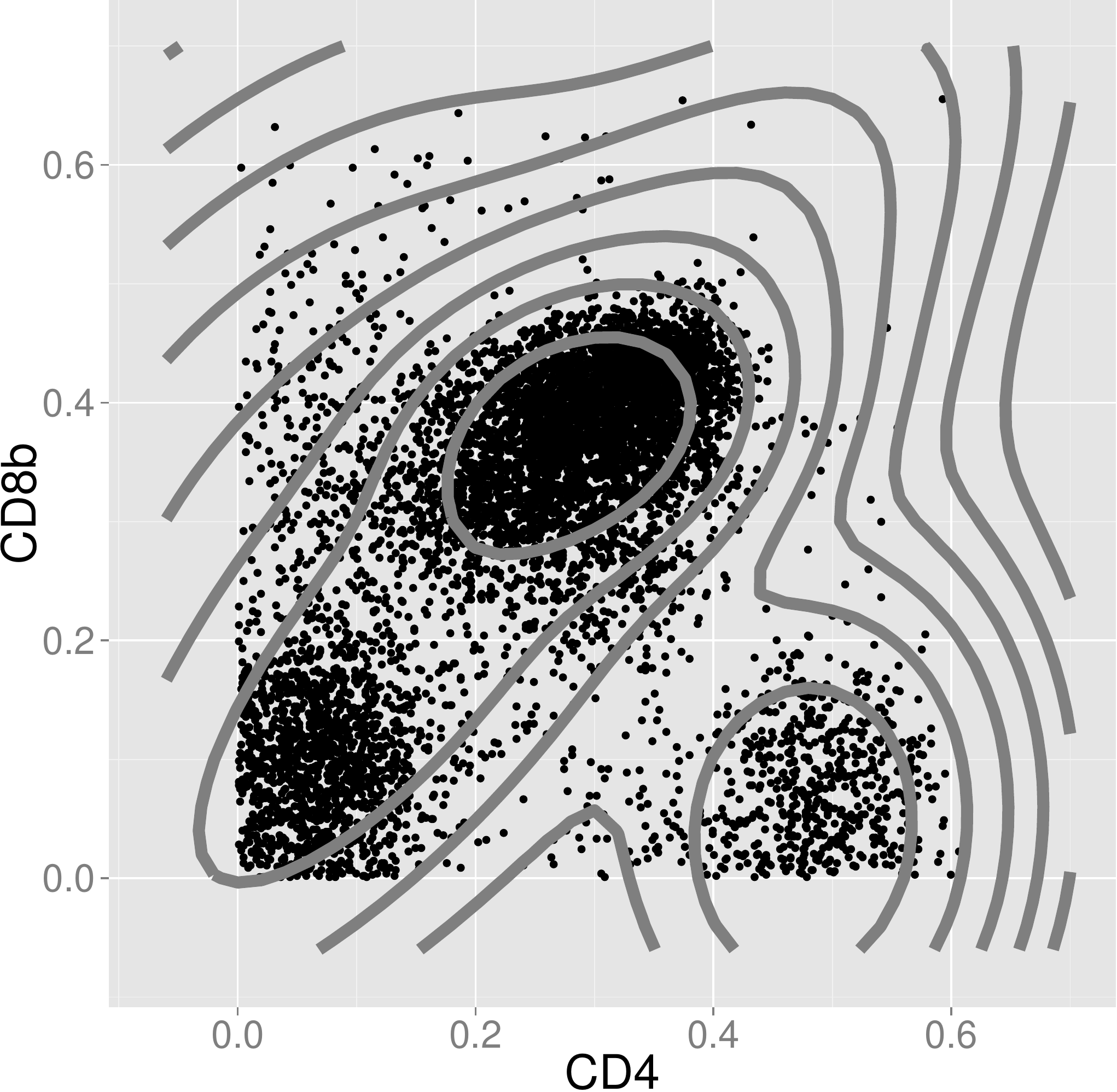}
    \includegraphics[width=.27\textwidth]{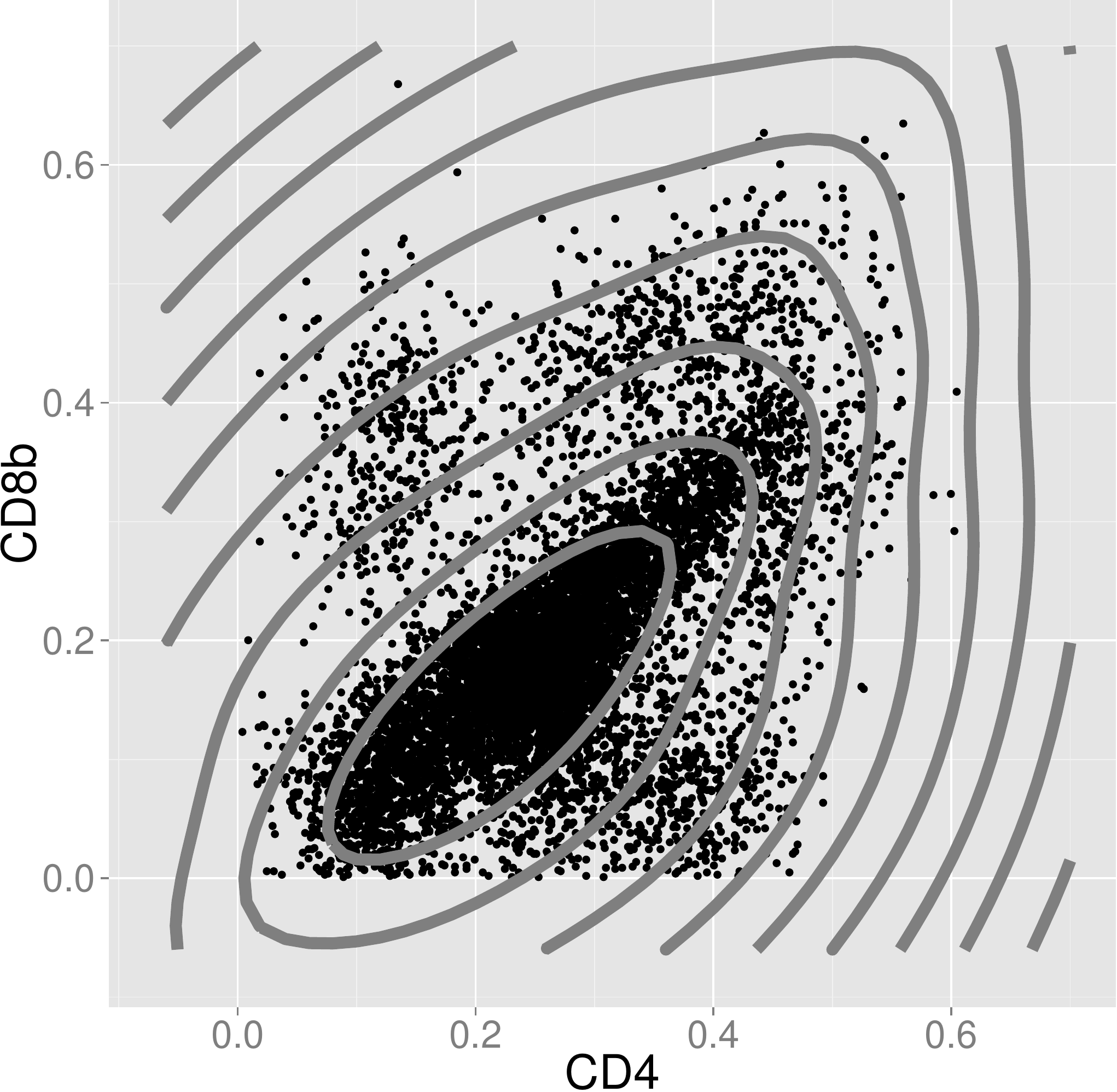}
\caption{ Scatterplots of first two dimensions for control (left) and positive (right) group. Contours represent
log posterior-mean densities under a Dirichlet process mixture.}
  \label{fig:plot_gvhd}
  \end{figure}

We apply our algorithm to a dataset of flow cytometry measurements from patients subjected to bone-marrow transplant \citep{Brink07}. This graft-versus-host disease
dataset consists of 6,809 control and 9,083 positive observations depending on whether donor immune cells attack host cells.
Each observation consists of four biomarker measurements truncated to lie between $0$ and $1024$, though more complicated truncation rules are often 
used according to operator judgement \citep{Lee2012}. We normalize and plot the first two dimensions corresponding to the markers CD4 and CD8b in 
Figure \ref{fig:plot_gvhd}.
Truncation complicates the clustering of observations into homogeneous groups, an important step in the flow-cytometry pipeline called
gating.
Consequently, \cite{Lee2012} propose an EM algorithm for a truncated mixture of Gaussians, which must be adapted if
different mixture components or truncation rules are used. 

We model the untruncated distribution for each group as a Dirichlet process mixture of Gaussians \citep{Lo1984}, with points
outside the four-dimensional unit hypercube discarded to form the normalized dataset.
The Dirichlet process mixture model is a flexible nonparametric prior over densities 
parametrized by a concentration parameter $\alpha$ and a base probability measure.
We set $\alpha = 1$, and for the base measure, which gives the distribution over cluster parameters, we use a normal-inverse-Wishart distribution. 
Given the rejected variables, we can use standard techniques to update a representation of the Dirichlet process. 
We follow the blocked-sampler of \cite{IshJam2001} based on the stick-breaking representation of the Dirichlet process, using a truncation level of 
$50$ clusters. 
This corresponds to updating $\theta$, step 2 in Algorithm~\ref{alg:rej_post}. 
Having done this, we discard the old rejected samples, and produce a new set by drawing from a $50$-component mixture of Gaussians
model, corresponding to step 1 in Algorithm~\ref{alg:rej_post}.

Figure~\ref{fig:plot_gvhd} shows the log mean posterior densities for the first two dimensions from 10,000 iterations. 
While the control group has three clear modes, these are much less pronounced in the positive group.
Directly modeling observations with a Gaussian mixture model obscured this by forcing modes forced away from the edges. 
One can use components with bounded support in the mixture model, such as a Dirichlet process mixture of Betas; however, these do not reflect the 
underlying data generation process, and are unsuitable when different groups have different
truncation levels. By contrast, it is easy to extend our modeling ideas to allow groups to share
components \citep{TehJorBea2006}, thereby allowing better identification of disease predictors.

Our sampler took less than two minutes to run 1,000 iterations, not much longer than a typical Dirichlet process sampler for 
datasets of this size. The average number of augmented points was 3,960 and 4,608 for the two groups. We study our  sampler more systematically 
in the next section, but this application demonstrates the flexibility and simplicity of our main idea.


\section{Bayesian inference for the matrix Langevin distribution} \label{sec:Bays_inf}

The Stiefel manifold $V_{p,d}$ is the space of all $d \times p$ orthonormal matrices, 
that is, $d \times p$ matrices $X$ such that $X^TX=I_p$, where $I_p$ is the $p \times p$  identity matrix.
When $p=1$, this is the $d-1$ hypersphere $S^{d-1}$, and when $p=d$, this is the space of all orthonormal matrices
$O(d)$.
Probability distributions on the Stiefel manifold play an important role in statistics, signal processing and machine learning, with applications ranging 
from studies of orientations of orbits of comets and asteroids to principal components analysis to the estimation of rotation matrices.  
The simplest such distribution is the matrix Langevin distribution,
an exponential-family distribution whose density with respect to the invariant Haar
volume measure \citep{Edelman98thegeometry} is
$\ml(X\mid F)=\etr(F^TX)/Z(F)$. 
Here $\etr$ is the exponential-trace, and $F$ is a $d\times p$ matrix. The normalization constant $Z(F)=\mathstrut_0F_1(\frac{1}{2}d, \frac{1}{4}F^TF)$  is the hypergeometric
function  with matrix arguments, evaluated at $\frac{1}{4}F^TF$ \citep{chikusebook}.
Let $F = G \bkappa H^T$ be the singular value decomposition of $F$, where $G$ and $H$ are $d \times p$ and $p \times p$ orthonormal matrices, and $\bkappa$ is
a positive diagonal matrix. 
We parametrize $\ml$ by $(G, \bkappa, H)$, and one can think of $G$ and $H$ as orientations, with $\bkappa$ controlling the
concentration in directions determined by these orientations.
Large values of $\bkappa$ imply concentration along the associated directions, while setting $\bkappa$ to zero gives the uniform distribution on
the Stiefel manifold. It can be shown \citep{khatri1977} that $\mathstrut_0F_1(\frac{1}{2}d, \frac{1}{4}F^TF) =  \mathstrut_0F_1(\frac{1}{2}d, \frac{1}{4}\bkappa^T\bkappa)$, so that
this depends only on $\bkappa$. We write it as $Z(\bkappa$).

In our Bayesian analysis, we place independent priors on 
$\bkappa, G$ and $H$.
The latter two lie on the Stiefel manifolds $V_{p,d}$ and $V_{p,p}$, and we place matrix Langevin priors $\ml(\cdot\mid F_0)$ and $\ml(\cdot\mid F_1)$ on 
these: we will see below that these are conditionally conjugate. 
We place independent $\text{Gamma}(a_0,b_0)$ priors on the diagonal elements of $\bkappa$. 
However, the difficulty in evaluating the normalization constant $Z(\bkappa)$ 
makes posterior inference over $\bkappa$ doubly intractable.  Thus, in a 2006 University of Iowa PhD thesis, Camano-Garcia 
keeps $\bkappa$ constant, while \cite{hoff2009jrssb} uses a first-order Taylor expansion of the intractable term to run an approximate sampling algorithm.
Below, we show how fully Bayesian inference can be carried out over this quantity as well.

\subsection{A rejection  sampling algorithm} \label{sec:prior_sim}
We first describe a rejection sampling algorithm from \cite{hoff2009} to sample from $\ml$.
For simplicity, assume $H$ is the identity matrix. In the general case, we simply rotate  the resulting draw 
by $H$, since if $X \sim \ml(\cdot\mid F)$, then $XH \sim \ml(\cdot\mid FH^T)$.
At a high level, the algorithm sequentially proposes vectors
from the matrix Langevin on the unit sphere: this is also called von Mises-Fisher distribution and is easy to simulate
\citep{wood1994}.
The mean of the $r${th} vector is column $r$ of $G$, $G_{[:r]}$, projected onto the nullspace of the earlier vectors, $N_r$.
This sampled vector is then projected back onto $N_r$ and normalized, and 
the process is repeated $p$ times. Call the resulting
distribution $\seq$; for more details, see Algorithm \ref{alg:rej_smplr} and \cite{hoff2009}.
{
\vspace{.1in}
\begin{algorithm}
\caption{Proposal $\seq(\cdot\mid G, \bkappa)$ for matrix Langevin distribution \citep{hoff2009}}\label{alg:rej_smplr}
\begin{tabular}{p{1.4cm}p{12.2cm}}
\textbf{Input:}  & Parameters $G,\bkappa$; write $G_{[:i]}$ for column $i$ of $G$, and $\kappa_i$ for element $(i,i)$ of $\bkappa$ \\
\textbf{Output:} & An output  $X \in V_{p,d}$; write $X_{[:i]}$ for column $i$ of $X$ \\
\hline
\end{tabular}
\begin{algorithmic}[1]
  \STATE Sample $X_{[:1]} \sim \ml(\cdot\mid \kappa_1 G_{[:1]})$. 
  For $r \in \{2,\cdots p\}$
  \begin{enumerate}[(a)]
    \item Construct $N_r$, an orthogonal basis for the nullspace of $\{X_{[:1]},\cdots X_{[:r-1]} \}$.
    \item Sample $z \sim \ml(\cdot\mid \kappa_r N^T_r G_{[:r]})$, and
    \item Set $X_{[:r]} = z^T N_r/ \|z^T N_r\| $.
  \end{enumerate}
\end{algorithmic}
\end{algorithm}
}
Letting $I_k(\cdot)$ be the modified Bessel function of the first kind,
$\seq$ is a density on the Stiefel manifold with
\begin{align}
  \seq(X\mid G, \bkappa) &= \left\{\prod_{r=1}^p \frac{ \|\kappa_r N^T_r G_{[:r]}/2 \|^{(d-r-1)/2 }}{ \Gamma(\frac{d-r+1}{2} ) I_{(d-r-1)/2}(\| \kappa_r N^T_r G_{[:r]} \|)} \right\} \etr(\bkappa G^T X).
\end{align}
{Write $D(X, \bkappa, G)$ for the reciprocal of the term in braces. Since
$I_k(x)/x^k$ is an increasing function of $x$, and  $\|N^T_r G_{[:r]}\| \le \|G_{[:r]}\| = 1$, we have the following bound $D(\bkappa)$ for $D(X, \bkappa, G)$:}
\begin{align}
 D(X,\bkappa, G) &\le \prod_{r=1}^p  \frac{ \Gamma(\frac{d-r+1}{2} ) I_{(d-r-1)/2}(\| \kappa_r \|)}{ \|\kappa_r/2 \|^{(d-r-1)/2 }} = D(\bkappa).\qquad \qquad \qquad \nonumber
\end{align}
This implies $\etr(\bkappa G^T X) \le D(\bkappa) \seq(X\mid G,\bkappa) $, allowing
the following rejection sampler: draw a sample $X$ from $\seq(\cdot)$, and accept with probability
$D(X, \bkappa, G)/D(\bkappa)$. The accepted proposals come from $\ml(\cdot\mid G,\bkappa)$, and for samples from $\ml(\cdot\mid G,\bkappa,H)$,
post multiply these by $H$.

\subsection{Posterior sampling}   \label{sec:post_sim}

Given a set of $n$ observations $\{X_i\}$, and writing $S = \sum_{i=1}^n X_i$, we have: 
\begin{align}
  p(G, \bkappa, H \mid X_i\}) & \propto {\etr(H \bkappa G^T S)p(H) p(G) p(\bkappa)}/{Z(\bkappa)^{n}}. \nonumber
\end{align}

At a high level, our approach is a Gibbs sampler that sequentially updates $H, G$ and $\bkappa$. 
The pair of matrices $(H,G)$ correspond to the tractable $\theta_1$ in Algorithm~\ref{alg:rej_post}, while $\bkappa$ corresponds to $\theta_2$.
Updating the first two is straightforward, while the third requires our augmentation scheme. \\
\hspace{.1in}\\
\noindent \textbf{1. Updating $G$ and $H$:} \label{sec:update_conj} 
{With a matrix Langevin prior on $H$, the posterior is }
\begin{align}
  p(H\mid X_i,\bkappa, G) & \propto \etr\left\{(S^T G \bkappa + F_0)^T H\right\}. \nonumber
\end{align}
This is just the matrix Langevin distribution over rotation matrices, and one can sample from this following Section \ref{sec:prior_sim}.
From here onwards, we will rotate the observations by $H$, allowing us to ignore this term. Redefining $S$ as $SH$, the posterior over $G$ is also
a matrix Langevin:
\begin{align}
  p(G\mid X_i\},\bkappa) & \propto \etr\left\{(S \bkappa + F_1)^T G\right\}. \nonumber
\end{align}

\noindent \textbf{2. Updating $\bkappa$:} 
Here, we exploit the rejection sampler scheme of the previous section, 
and instantiate the rejected proposals using Algorithm \ref{alg:rej_sim}. 
From Section \ref{sec:prior_sim}, the joint probability is 
\begin{align}
  p(\{X_i, \cY_i\}\mid G,\bkappa)    
         &= \frac{ \etr\left\{\bkappa G^T\left(S+\sum_{j=1}^{|\cY_i|}  Y_{ij}\right ) \right\} }{D(\bkappa)^{1+|\cY|}}
             \prod_{i=1}^n   \prod_{j=1}^{|\cY|} \frac{ \left\{D(\bkappa) -  D(Y_{ij}, G, \bkappa)\right\}}{D(Y_{ij}, G, \bkappa)}.  \label{eq:rej_joint1}
\end{align}
All terms in the expression above can be evaluated easily, allowing 
a simple Metropolis-Hastings algorithm in this augmented space.
In fact, we can go further, calculating gradients to run a Hamiltonian Monte Carlo algorithm
\citep{Neal2010} that makes significantly more efficient proposals than a random-walk sampling algorithm. 
In particular,
let $N = n + \sum_{i=1}^n |\mathcal{Y}_i|$, and
$S = \sum_{i=1}^n(X_i + \sum_{j=1}^{|\mathcal{Y}_i|} Y_{ij})$. The log joint probability $L \equiv \log\left\{p(\{X_i,\cY_i\})\right\}$ is
\begin{align}
 L &= \text{trace}(G^T \bkappa S) +\sum_{i=1}^n \sum_{j=1}^{|\cY_i|}\left[ \log \left\{D(\bkappa) - D(Y_{ij}, \bkappa) \right\}\right. 
                        - \left. \log\{D(Y_{ij}, \bkappa)\}\right] - n \log\left\{D(\bkappa)\right\}. \nonumber
\end{align}
Writing $D(Y, \bkappa) =
 \left\{C\prod_{r=1}^p \frac{ I_{(d-r-1)/2}(\| \kappa_r N^T_r G_r \|)}{ \|\kappa_r N^T_r G_r \|^{(d-r-1)/2 }} \right\}$ as $C\tD(Y, \bkappa)
$, Appendix \ref{sec:gradient} shows that 
\begin{align}
\frac{\dif L }{\dif \kappa_k} 
    &=  G_{[,k]}^T S_{[,k]} +\sum_{i=1}^n \sum_{j=1}^{|\cY_i|}\left[\frac{ \frac{I_{(d-k+1)/2}}{I_{(d-k-1)/2}}(\kappa_k)  - N^T_kG_k \frac{I_{(d-k+1)/2}}{I_{(d-k-1)/2}}(\kappa_k N^T_kG_k ) }
                             { \left\{1 - \frac{\tD(Y_{ij}, \bkappa)}{\tD(\bkappa)}\right\}}  \right]
         - N  \frac{I_{(d-k+1)/2}}{I_{(d-k-1)/2}}(\kappa_k). \nonumber
\end{align}

We use this gradient to construct a Hamiltonian Monte Carlo sampler \citep{Neal2010} for $\bkappa$. 
Here, it suffices to note that a proposal involves taking $L$ 
`leapfrog' steps of size $\epsilon$ along the gradient, and accepting the resulting state with probability proportional to the product of equation
\eqref{eq:rej_joint1}, and a simple Gaussian `momentum' term. The acceptance probability depends on how well the
$\epsilon$-discretization approximates the continuous dynamics of the system, and choosing a small $\epsilon$ and a large $L$ can give global moves with high
acceptance probability. A large $L$ however comes at the cost of a large number of gradient evaluations. We study this trade-off
in Section \ref{sec:Bayes_expt}. 

\subsection{Vectorcardiogram dataset}  \label{sec:vcg_np}

The vectorcardiogram is a loop traced by the {cardiac vector} during a cycle of the heart beat. The two directions of orientation of this
loop in three-dimensions form a point on the Stiefel manifold. The dataset of \cite{vecdata} includes $98$ such recordings,
and is displayed in the left subplot of Figure \ref{fig:vcg_data}. We represent each observation with a pair of orthonormal vectors, with the set of 
cyan lines to the right forming
the first component. This empirical distribution 
possesses a single mode, so that
the matrix Langevin distribution appears a suitable model.

  \begin{figure}
  \centering
  \begin{minipage}{0.41\textwidth}
    \includegraphics[width=.9\textwidth]{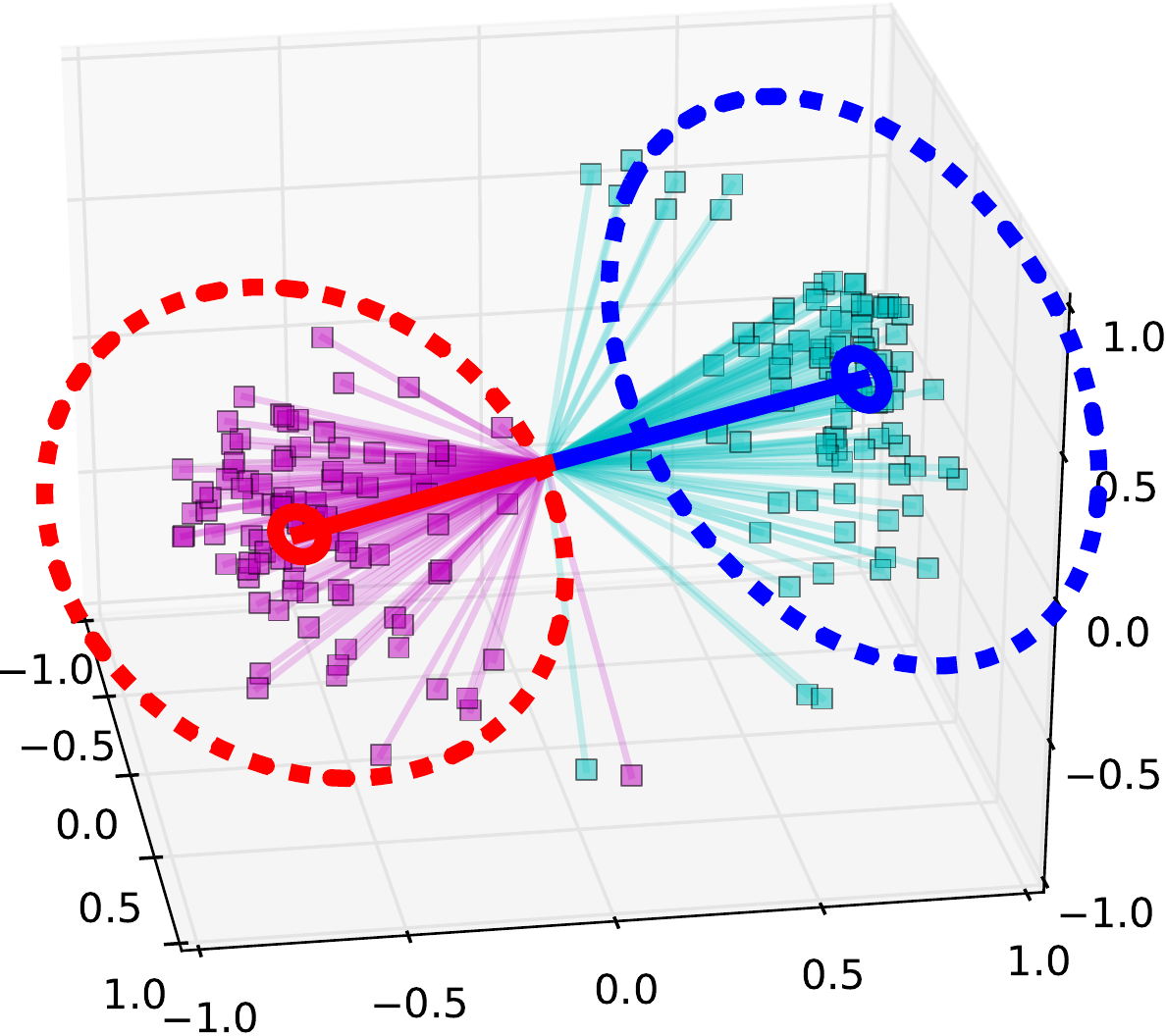}
  \end{minipage}
  \begin{minipage}{0.31\textwidth}
  \centering
    \includegraphics[width=.7\textwidth]{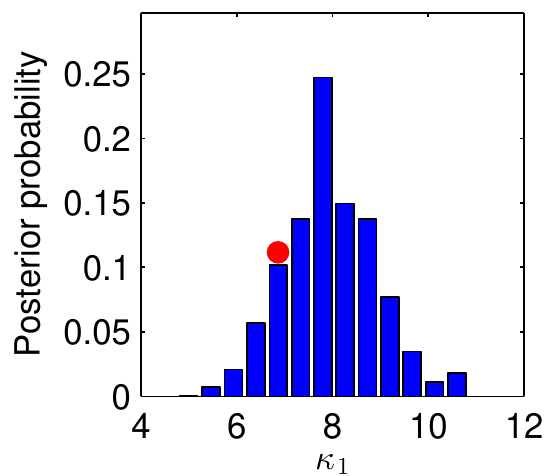}
  \centering
    \includegraphics[width=.7\textwidth]{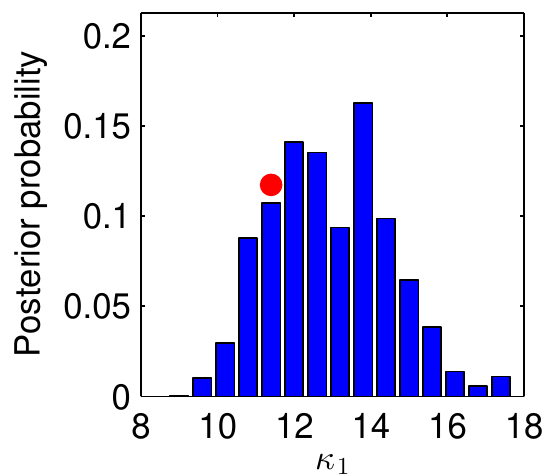}
  \end{minipage}
\caption{(Left) Vector cardiogram dataset with inferences. Bold lines are maximum likelihood estimates of $G$, and solid circles
 contain $90\%$ posterior mass. Dashed circles are $90\%$ predictive probability regions.
 (Right) Posterior over $\kappa_1$ and $\kappa_2$, circles are maximum likelihood estimates.}
  \label{fig:vcg_data}
  \end{figure}

We place weak independent exponential priors with mean $10$ and variance $100$ on the scale parameter $\bkappa$, and a uniform prior on the location
parameter $G$. We restrict $H$ to be the identity matrix. Inferences were carried out using the Hamiltonian sampler to produce 10,000 samples, with a
burn-in period of 1,000. For the leapfrog dynamics, we set a step size of $0.3$, with the number of steps equal to $5$.
We fix the `mass parameter' to the identity matrix as is typical.
We implemented all algorithms in R, building on
code from the $\mathtt{rstiefel}$ package of Peter Hoff. All
simulations were run on an Intel Core 2 Duo 3 Ghz CPU.
For comparison, we include the maximum likelihood estimates of $\bkappa$ and $G$. 
For  $\kappa_1$ and $\kappa_2$, these were $11.9$ and $5.9$, and we plot these in the right half of Figure \ref{fig:vcg_data}
as the red circles. 

The bold straight lines in Figure \ref{fig:vcg_data} (left) show the maximum likelihood estimates of the components of $G$, with the small
circles corresponding to $90\%$ Bayesian credible regions
estimated from the Monte Carlo output.
The dashed circles correspond to $90\%$ predictive probability regions for the Bayesian model. For these, we generated  $50$ points on $V_{3,2}$ for 
each sample, with parameters specified by that sample. The dashed circles contain $90\%$ of these points across all samples.
Figure \ref{fig:vcg_data} (right) show the posterior over $\kappa_1$ and $\kappa_2$.

\subsection{Comparison of exact samplers} \label{sec:Bayes_expt}

To quantify sampler efficiency, 
we estimate the
effective sample sizes produced per unit time. This corrects for
correlation between successive Markov chain samples
by estimating the number of independent samples produced; for this
we used the $\mathtt{rcoda}$ package of  \cite{Rcoda2006}.

  \begin{figure}
  \centering
    \includegraphics[width=.26\textwidth]{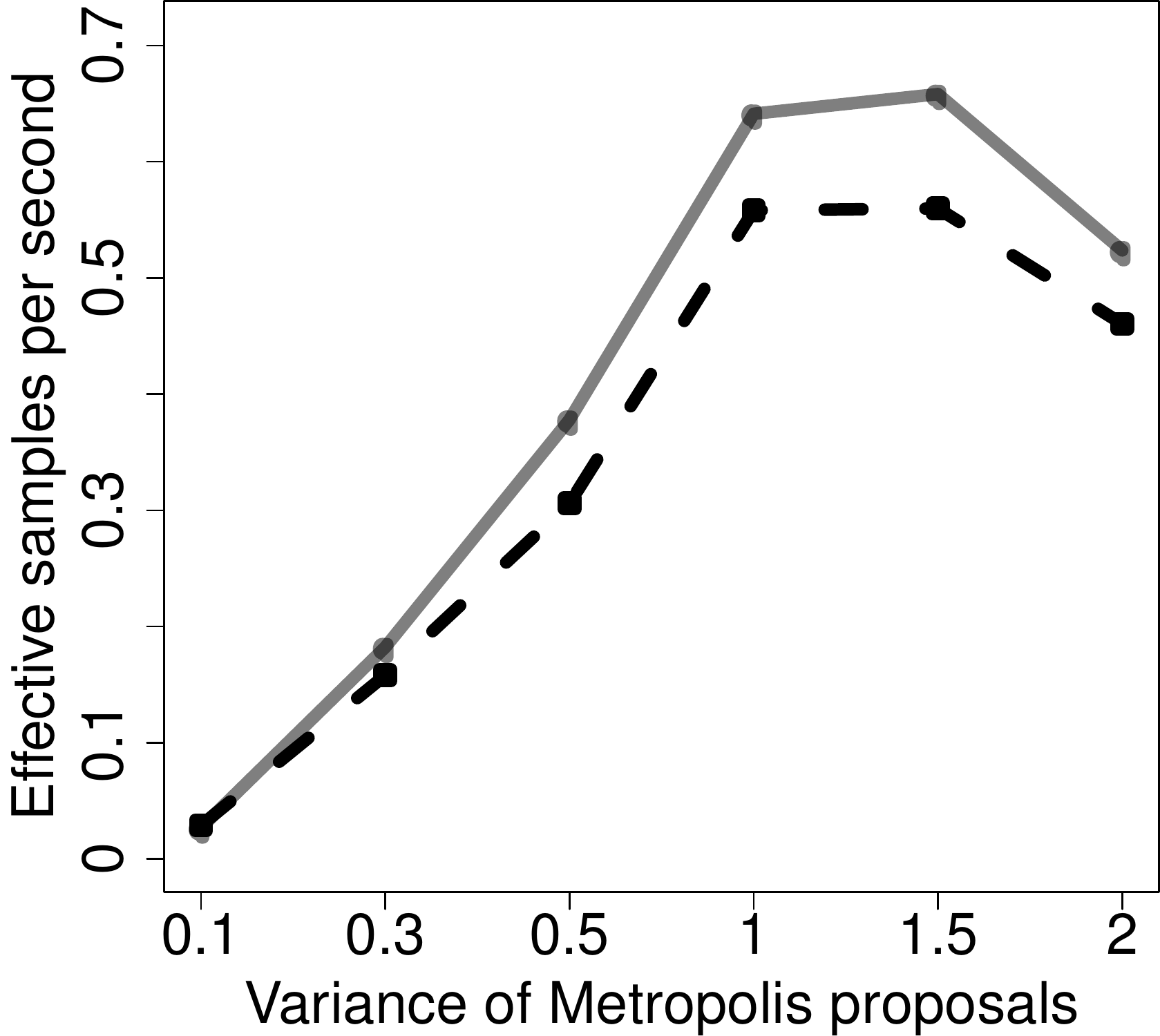}
  \centering
    \includegraphics[width=.26\textwidth]{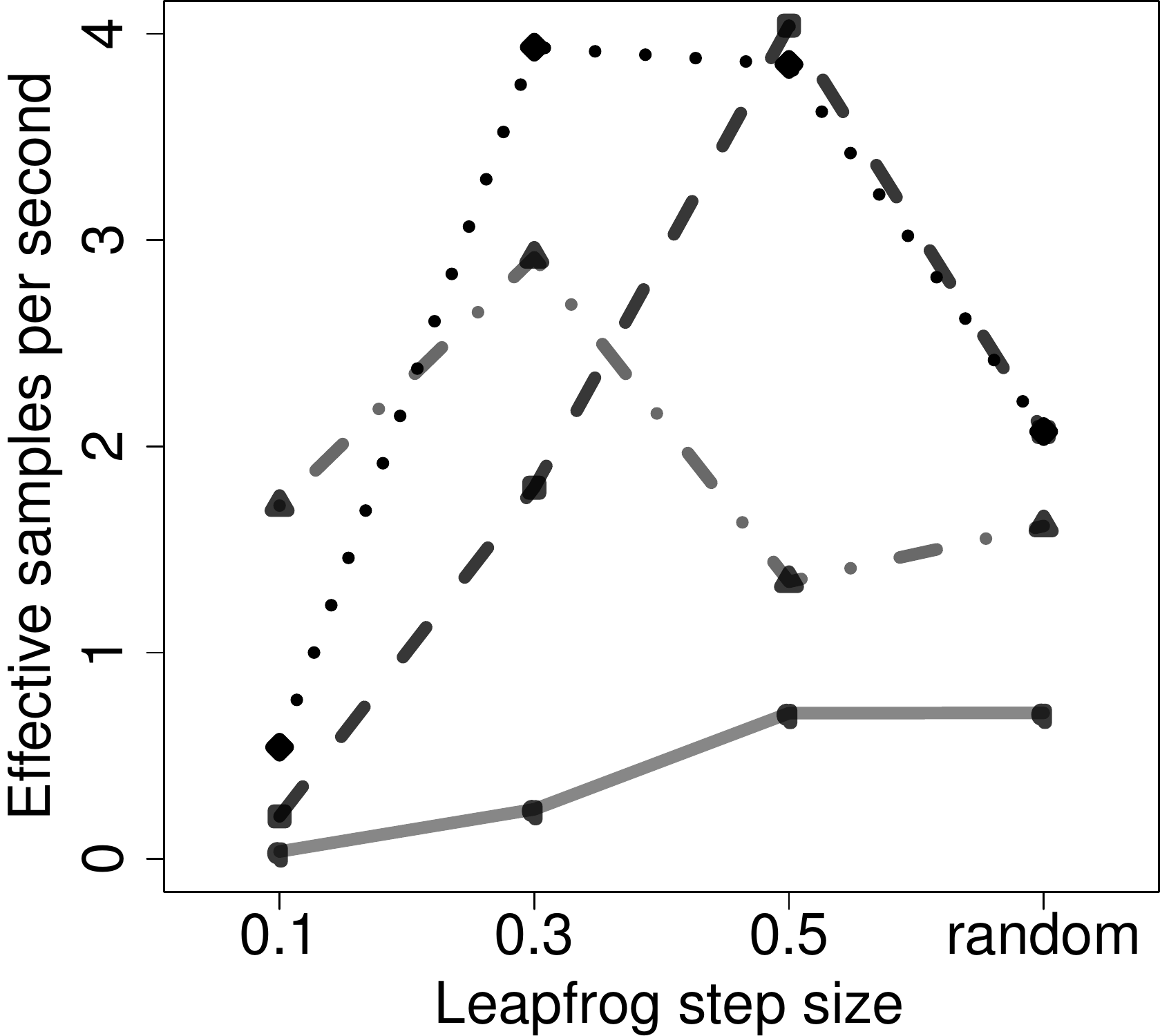}
\caption{Effective samples per second for (left) random walk and (right) Hamiltonian samplers. From bottom to top at abscissa $0.5$: (left) 
 Metropolis-Hastings data-augmentation sampler and exchange sampler, and (right) 1/10/5/3 leapfrog steps of Hamiltonian sampler.}
  \label{fig:samplers_comp}
  \end{figure}

The left plot in Figure \ref{fig:samplers_comp} considers two
Metropolis-Hastings samplers, the exchange sampler and our latent variable sampler on the vectorcardiogram dataset.
Both samplers
perform a random walk in the $\bkappa$-space, with the steps drawn for a normal
distribution whose variance increases along the horizontal axis.
The vertical axis shows the
median effective sample size per second for the components of $\bkappa$.
The figure shows that both samplers' performance
peaks when the proposals have a variance between $1$ and $1.5$, with the
exchange sampler performing slightly better. However, 
the real advantage of our sampler is that introducing the latent variables
results in a joint distribution without any intractable terms,
allowing the use of more sophisticated sampling algorithms. The plot to the right studies the Hamiltonian
Monte Carlo sampler described at the end of Section \ref{sec:latent_hist}. Here we
vary the size of the leapfrog steps along the horizontal axis, with the different
curves corresponding to different numbers of leapfrog steps.
We see that this performs an order of magnitude better than either of the previous
algorithms, with performance peaking with $3$ to $5$ steps of size $0.3$ to
$0.5$, fairly typical values for this algorithm. This shows the advantage of exploiting
gradient information in exploring the parameter space.

\subsection{Comparison with an approximate sampler}
In this section, we consider an approximate sampler based on an asymptotic approximation to $Z(\bkappa)= \mathstrut_0F_1(\frac{1}{2}d, \frac{1}{4}\bkappa^T\bkappa)$ for large values of
$(\kappa_1, \cdots, \kappa_n)$ \citep{khatri1977}:
\begin{align}
  Z(\bkappa)  \simeq &
         \left\{ \frac{2^{-\frac{1}{4}p(p+5) + \frac{1}{2}pd}}{\pi^{\frac{1}{2}p}} \right\} \etr(\bkappa) \prod_{j=1}^p \Gamma\left(\frac{d-j+1}{2} \right)
   \left[ \left\{\prod_{j=2}^p \prod_{i=1}^{j-1}(\kappa_i + \kappa_j)^{\frac{1}{2}} \right\} \prod_{i=1}^p \kappa_i^{\frac{1}{2}(d-p)} \right]^{-1}. \nonumber
\end{align}
We use this approximation in the acceptance probability
of a Metropolis-Hastings algorithm; it can similarly be used to construct a Hamiltonian sampler. 
For a more complicated but accurate approximation, see \cite{Kume:2013:SAN}. In general however, using
such approximate schemes involves the ratio of two approximations, and can
have very unpredictable performance. 
  \begin{figure}
  \centering
    \includegraphics[width=.25\textwidth]{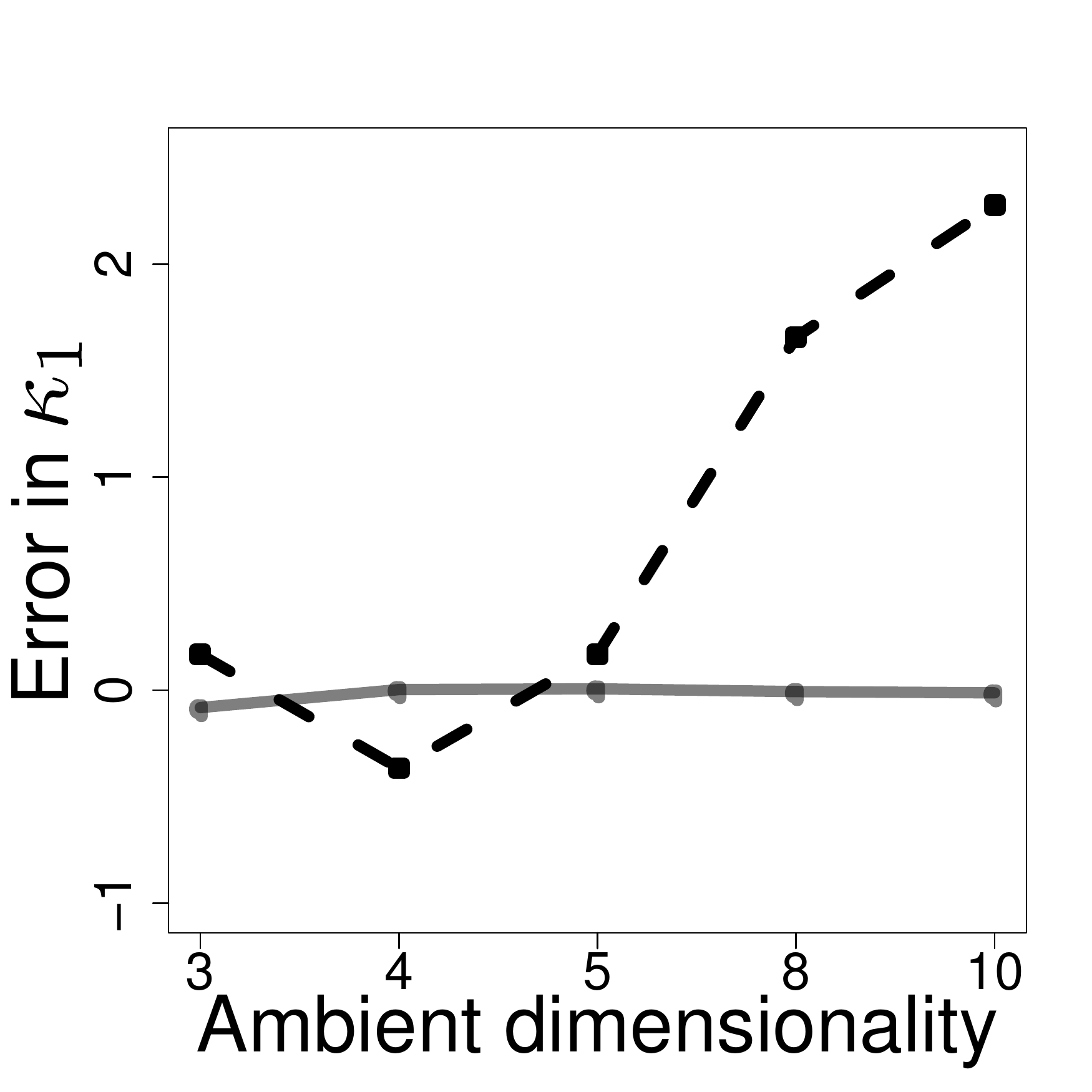}
  \centering
    \includegraphics[width=.25\textwidth]{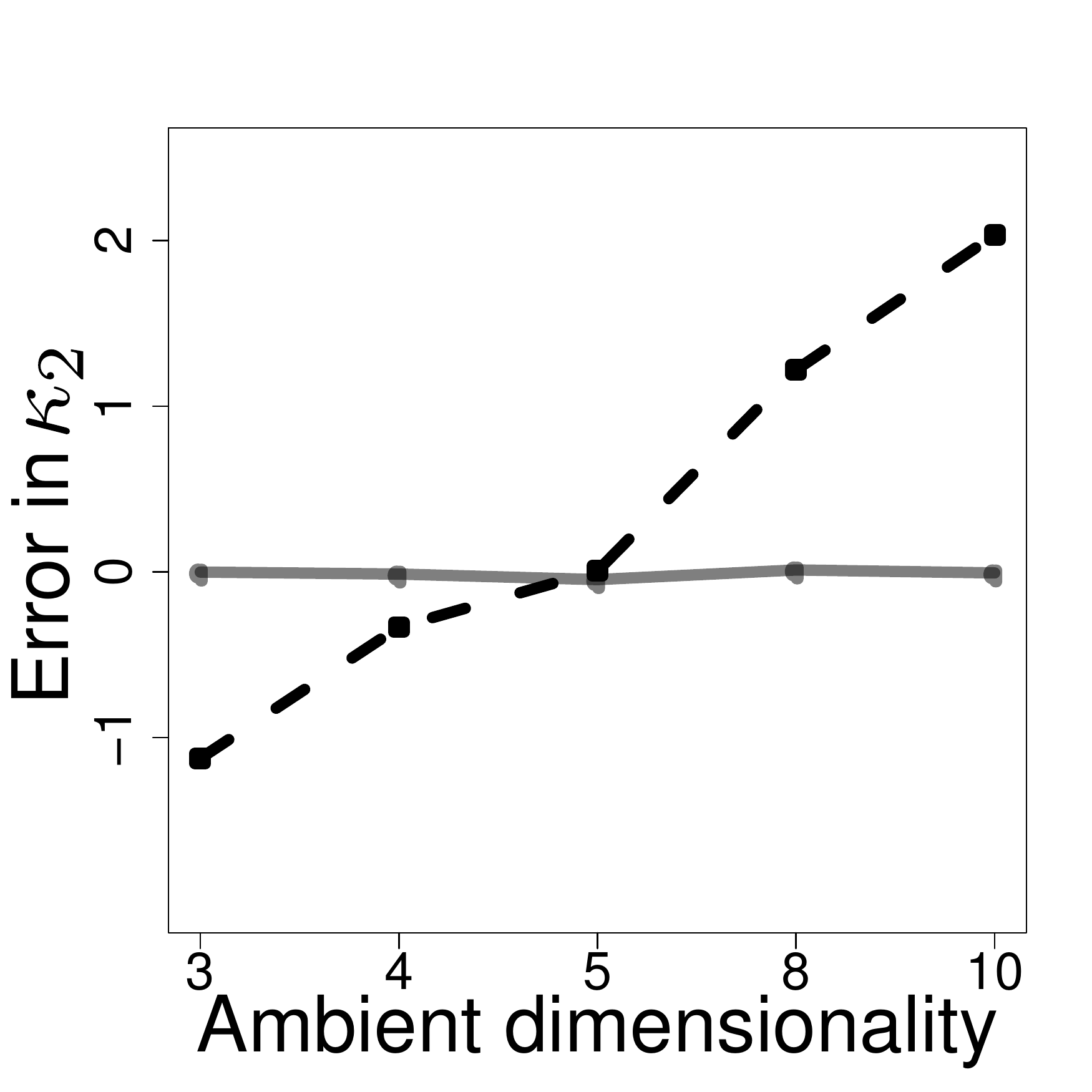}
  \centering
    \includegraphics[width=.25\textwidth]{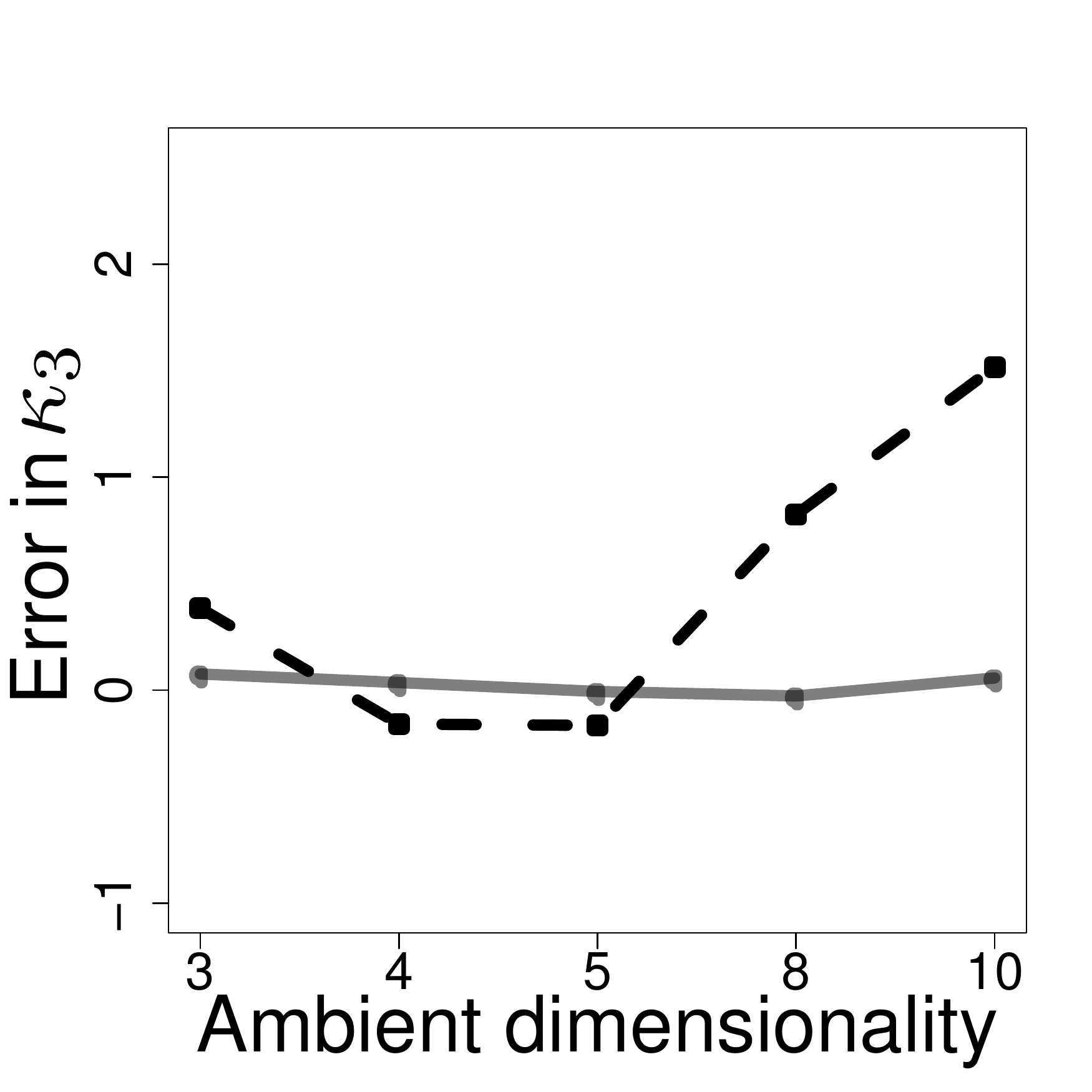}
\caption{Errors in the posterior mean. Solid/dashed lines are Hamiltonian/approximate sampler.}
  \label{fig:approx_comp}
  \end{figure}

On the vectorcardiogram dataset, the approximate sampler 
is about forty times faster than the exact samplers. For larger
datasets, this difference will be even greater, and the real question is how accurate
the approximation is. Our exact sampler allows us to study this: 
we consider the Stiefel manifold $V_{d,3}$,
with the three diagonal elements of $\bkappa$ set to $1, 5 $ and $10$. With this setting of $\bkappa$, and
a random $G$, we generate datasets with $50$ observations with $d$ taking values $3, 4, 5, 8,$ and $10$.
In each case, we estimate the posterior mean of $\bkappa$ by running the exchange sampler, and treat this as the
truth. We compare this with posterior means returned by our Hamiltonian sampler, as well as the approximate sampler.  Figure \ref{fig:approx_comp}
shows these results, with the three subplots corresponding to the three components of $\bkappa$, and ambient dimensionality
$d$ increasing along the horizontal axis.
As expected, the two exact samplers agree, and the Hamiltonian sampler has almost no `error'.
The approximate sampler is more complicated. For values of $d$ around $5$, its estimated
posterior mean is close to that of the exact samplers. Smaller values lead to an approximate posterior mean that underestimates the actual posterior mean,
while in higher dimensions, the opposite occurs. Recalling that $\bkappa$ controls the concentration of the matrix Langevin distribution about its
mode, this implies that in high dimensions, the approximate sampler underestimates uncertainty in the distribution of future observations.


\section{The Gaussian process density sampler}\label{sec:gpds}
Our next application is the Gaussian process density sampler of \cite{adams_gpds},
a nonparametric prior for probability densities induced by a logistic transformation of a random function from a Gaussian process.  
Letting $\sigma(\cdot)$ denote the logistic function, the random density is 
\begin{align}
g(x) &\propto g_0(x) \sigma\{f(x)\},\qquad
f \sim  \mbox{\small{GP}}, \nonumber 
\end{align}
with $g_0(\cdot)$ a parametric base density and $\mbox{\small{GP}}$ denoting a Gaussian process.
The inequality
$g_0(x) \sigma\{f(x)\} \le g_0(x)$
allows a rejection sampling algorithm 
by making proposals from $g_0(\cdot)$.
At a proposed location $x^*$, we sample the function value $f(x^*)$ conditioning on 
all previous evaluations, and accept the proposal with probability 
$\sigma\{f(x^*)\}$.  Such a scheme involves no approximation error, and only requires evaluating the random function on a finite set of points.
Algorithm \ref{alg:gpds} describes the steps involved in generating $n$ observations.

{
\newcommand{\WP}[1]{\STATE \textbf{With probability} (#1)}

\begin{algorithm}
\caption{Generate $n$ new samples from the Gaussian process density sampler}\label{alg:gpds}
\begin{tabular}{p{1.4cm}p{12.2cm}}
\textbf{Input:}  & A base probability density $g_0(\cdot)$, \\
                 & Previous accepted and rejected proposals $\tilde{X}$ and $\tilde{Y}$, \\
                 & Gaussian process evaluations $f_{\tilde{X}}$ and $f_{\tilde{Y}}$ at these locations. \\
\textbf{Output:} & $n$ new samples $X$, with the associated rejected proposals $Y$, \\ 
                 & Gaussian process evaluations $f_X$ and $f_Y$ at these locations. \\
\hline
\end{tabular}
\begin{algorithmic}[1]
  \REPEAT
  \STATE Sample a proposal $y$ from $g_0(\cdot)$.
  \STATE Sample $f_y$, the Gaussian process evaluated at $y$, conditioning on $f_X$, $f_Y$, $f_{\tilde{X}}$ and $f_{\tilde{Y}}$. 
  \STATE \textbf{with probability }{$\sigma(f_y)$}: \quad Accept $y$ and add it to $X$. Add $f_y$ to $f_X$.
  \STATE $\quad$ \textbf{else}:\qquad \qquad \qquad  \qquad Reject $y$ and add it to $Y$. Add $f_y$ to $f_Y$.
  \UNTIL{$n$ samples are accepted.}
\end{algorithmic}
\end{algorithm}
}

\subsection{Posterior inference}
Given observations $X = \{x_1, \cdots, x_n\}$,
we are interested in $p(g\mid X)$, the posterior over the underlying density. 
Since $g$ is determined by the modulating function $f$, we focus on 
$p(f\mid X)$. 
While this quantity is doubly intractable,
after augmenting the state space to include the proposals $\cY$ from the rejection sampling algorithm, 
 $p(f\mid X,\cY)$ has density with respect to the Gaussian process prior given by
$\prod_{i=1}^n \sigma\left\{f(x_i)\right\} \prod_{i=1}^{|\cY|} \left[1-\sigma\left\{f(y_i)\right\}\right]$, see also \cite{adams_gpds}.
In words, the posterior over $f$ evaluated at $X \cup \cY$ is just the posterior from a Gaussian process
classification problem with a logistic link-function, and with the accepted and rejected proposals corresponding to the two classes.
There are a number of Markov chain Monte Carlo methods such as Hamiltonian Monte Carlo or elliptical slice sampling 
\citep{murray2010} that are applicable in such a situation. Given $f$ on $X \cup \cY$, it can be evaluated anywhere else by conditionally sampling from a multivariate normal.

Sampling the rejected proposals $\cY$ given $X$ and $f$ is straightforward by Algorithm \ref{prop:rej_post}:
run the rejection sampler until $n$ accepts, and treat the rejected proposals generated
along the way as $\cY$. In practice, we do not have access to the entire function $f$, only its
values evaluated on $X$ and $\cY_{old}$, the location of the previous thinned variables. However, just as under the generative 
mechanism, we can retrospectively evaluate the function $f$ where needed.
After proposing from $g_0(\cdot)$, we sample the value of the function at this location conditioned on all previous evaluations, and
use this value to decide whether to accept or reject. We outline the inference algorithm in
Algorithm \ref{alg:gpds_mcmc}, noting that it is much simpler than that proposed in \cite{adams_gpds}.
We also refer to that paper for limitations of the exchange sampler in this problem.
{
\begin{algorithm}
\caption{A Markov chain iteration for inference in the Gaussian process density sampler}\label{alg:gpds_mcmc}
\begin{tabular}{p{1.4cm}p{12.2cm}}
\textbf{Input:}  & Observations $X$ with corresponding function evaluations $\tilde{f}_X$, \\
                 & Current rejected proposals $\tilde{Y}$ with corresponding function evaluations $\tilde{f}_{\tilde{Y}}$. \\
\textbf{Output:} & New rejected proposals $Y$, \\ 
                 & New Gaussian process evaluations $f_X$ and $f_Y$ at $X$ and $Y$, \\
                 & New hyperparameters. \\
\hline
\end{tabular}
\begin{algorithmic}[1]
  \STATE Run Algorithm \ref{alg:gpds} to produce $|X|$ accepted samples, with $X, \tilde{Y}, \tilde{f}_X$ and $\tilde{f}_{\tilde{Y}}$ as inputs.
  \STATE Replace $\tilde{Y}$ and $f_{\tilde{Y}}$ with values returned by the previous step; call these $Y$ and $\hat{f}_Y$.
  \STATE Update $\tilde{f}_X$ and $\hat{f}_Y$ using for example, hybrid Monte Carlo, to get $f_X$ and $f_Y$.
  \STATE Update Gaussian process and base-distribution hyperparameters.
\end{algorithmic}
\end{algorithm}
}

\subsection{Experiments}   \label{sec:gpds_expt}

  \begin{figure}
  \centering
    \includegraphics[width=.35\textwidth]{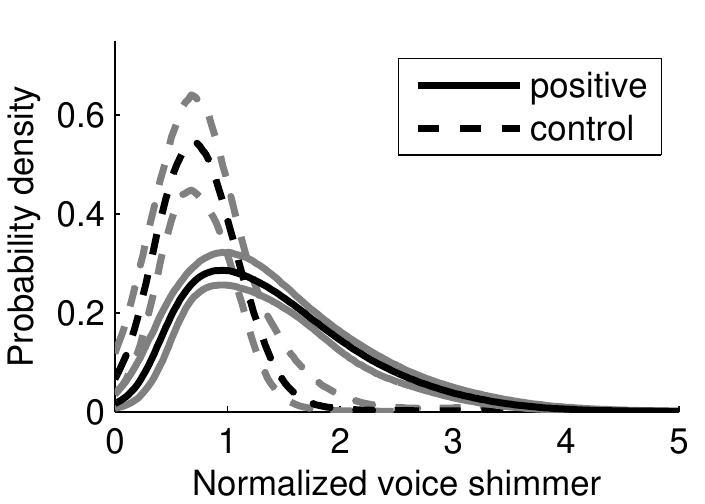}
    \includegraphics[width=.35\textwidth]{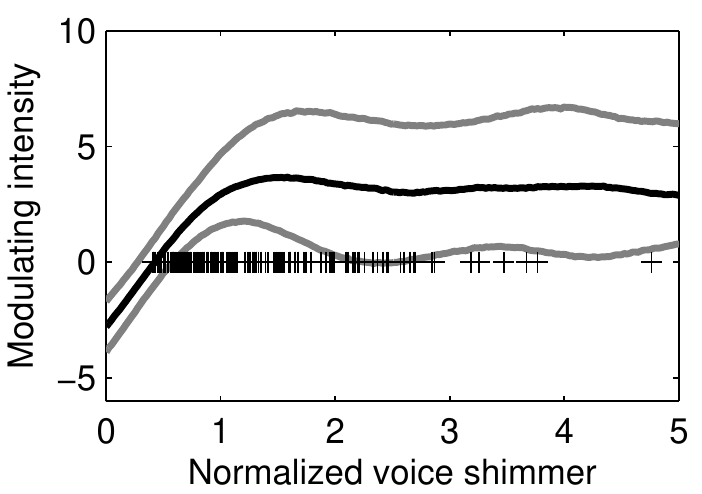}
\caption{(Left) Posterior density for positive and control groups, (Right) Posterior over the Gaussian process function for
positive group with observations. Both plots show the median with $80$ percent posterior credible intervals.}
  \label{fig:plot_glx}
  \end{figure}

  Voice changes are a symptom and measure of onset of Parkinson's disease, and one attribute is voice shimmer, a measure of variation in
  amplitude. We consider a dataset of such measurements for subjects with and without the disease \citep{little07}, with
$147$ measurements with, and $48$ without the disease.
We normalized these to vary from $0$ to $5$, and  used the model of \cite{adams_gpds} as a prior on the underlying probability densities. 
We set $g_0(\cdot)$ to a normal $\mathcal{N}(\mu,\sigma^2)$, 
with a normal-inverse-Gamma prior on $(\mu, \sigma)$. The latter had parameters $(0,.1,1,10)$. The Gaussian process had a 
squared-exponential kernel, with variance and length-scale of $1$.
For each case, we ran a Matlab implementation of our data augmentation algorithm to produce 2,000 posterior samples after a burn-in of $500$ samples. 

  \begin{figure}
  \centering
    \includegraphics[width=.35\textwidth]{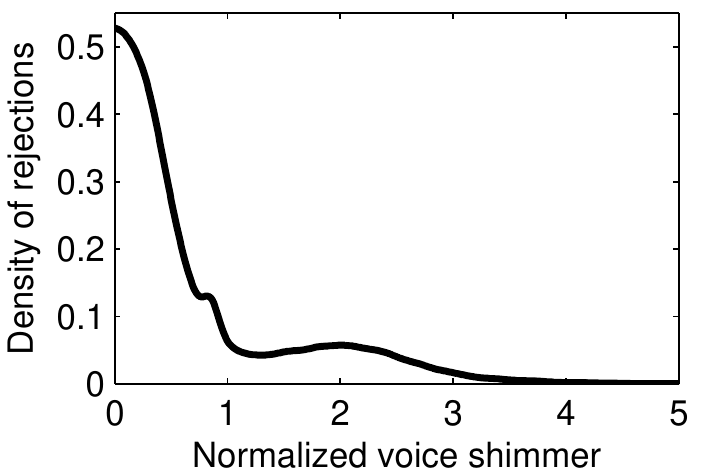}
    \includegraphics[width=.35\textwidth]{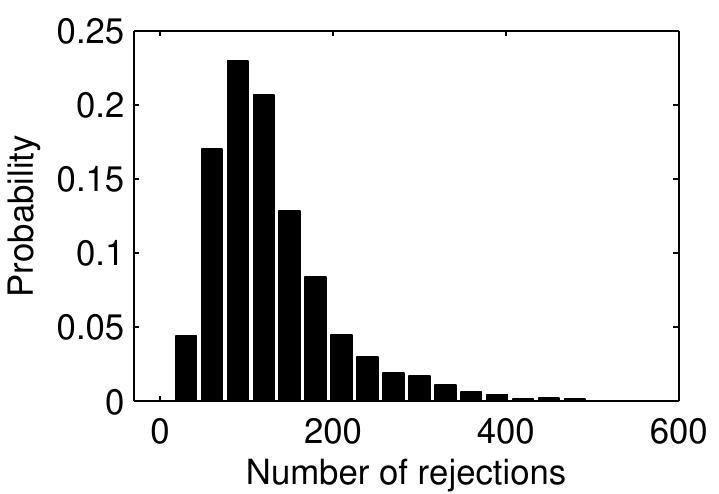}
\caption{(Left) Kernel density estimate of locations of rejected proposals and (Right) Histogram of number of rejected proposals for positive group.}
  \label{fig:plot_glx2}
  \end{figure}

  Figure~\ref{fig:plot_glx} (left) shows the resulting posterior over densities, corresponding to $\theta$ in Algorithm~\ref{alg:rej_post}.
The control group is fairly Gaussian, while the disease group is skewed to the right.
The right plot in the same figure focuses on the deviation from normality by plotting the posterior over the latent function $f$. 
We see that to the right of $0.5$, it is larger than its prior mean of $0$, implying larger probability than under a Gaussian.
Figure~\ref{fig:plot_glx2} studies the distribution of the rejected proposals $\cY$.
The left plot shows the distribution of their locations:
most of these occured near the origin. Here, the disease density reverts to Gaussian or even sub-Gaussian, with the intensity function taking small 
values.  The right plot is a histogram of the number of rejected proposals: this is typically 
around $100$ to $150$, though
the largest value we observed was $668$. Since inference over the latent function involves evaluating it at the locations of the accepted as
well as rejected proposals, the largest covariance matrix we had to deal with was about $600 \times 600$; typical values were around $100 \times 100$. Using
the same setup as Section~\ref{sec:Bayes_expt}, it took a na\"ive Matlab implementation $26$ and $18$ minutes to run 2,500 iterations for the
disease and control datasets. One can imagine 
computations becoming unwieldy for a large number of observations, or when there is large mismatch between the true density and the base-measure $g_0(\cdot)$. 
In such situations, one might have to choose the Gaussian process covariance kernel more carefully, use one of many sparse approximation techniques,
or use other nonparametric priors like splines instead. In all these cases, we can use our algorithm to recover the rejected proposals $\cY$, 
and given these, posterior inference over $f$ 
can be carried out using standard techniques.


\section{Discussion}\label{sec:conc}
We described a simple approach to carry out Markov chain Monte Carlo inference when data generation
involves a rejection sampling algorithm. 
Our algorithm is simple and efficient, and allows us to exploit ideas
like Hamiltonian Monte Carlo to carry out efficient inference. While our algorithm is exact, 
it also provides a framework for faster, approximate algorithms.
For instance, the number of rejected proposals preceeding any observation is a random number that
a priori is unbounded. One can bound the computational cost of an iteration by limiting the maximum number
of rejected proposals. Similarly, one might try sharing rejected proposals across observations. We leave
the study of the approximate Markov chain algorithms resulting from such `user impatience' for future research. Also left open is a more careful analysis
of Markov mixing rates for the applications we considered. There are also a number of potential applications that we have
not described here: particularly relevant are rejection samplers for stochastic differential
equations \citep{beskos05,bladt2014} .


\section{Acknowledgement}
This work was supported by the National Institute of Environmental Health Sciences of the National Institute of Health.

\appendix
\section{Proofs}



\begin{prf1}[of Proposition~\ref{prop:rej_post}]

  Rejection sampling first proposes from $q(x)$, and then accepts with probability $f(x)/\{Mq(x)\}$. 
  Conceptually, one can first decide whether to accept or reject, and then conditionally sample the location.
  The marginal acceptance probability is $Z(\theta)/M$, the area under $f(\cdot,\theta)$ divided by that under $M q(\cdot\mid\theta)$.
  An accepted sample $x$ is distributed as the target distribution $f(x, \theta)/Z(\theta)$, while rejected samples are distributed as 
  $\frac{Mq(x\mid\theta) - f(x,\theta)}{M-Z(\theta)}$. This two component mixture is just the proposal $q(x)$.
  While this scheme loses the computational benefits that motivate the original algorithm, it shows that the location of an accepted sample is independent
  of the past, and consequently, that the number and locations of rejected samples preceding an accepted sample is independent of the location
  of that sample. Consequently, one can use the rejected samples preceding any other accepted sample.
\end{prf1}

\begin{prf1}[of Theorem~\ref{thrm:conv_rate}]
It follows easily from Bayes' rule that for an observation $X$,
    $$p(\theta|X,\cY) \ge p(\theta|X) \frac{b_f}{B_f} \left(\frac{b_qr}{B_q}\right)^{|\cY|}.$$
 Let the number of observations $|X|$ be $n$. Then,
\begin{align}
  k(\htheta\mid \theta) &= \int_{\bU^n} p(\htheta\mid \cY,X) p(\cY\mid \theta,X) \mathrm{d}\cY \nonumber \\
  & \ge  \left(\frac{b_f}{B_f}\right)^n p(\htheta\mid X) \prod_{i=1}^n \int_{\bU} \beta^{|\cY_i|} p(\cY_i\mid \theta,X)  \mathrm{d}\cY_i \nonumber \\
    & =  \left(\frac{b_f}{B_f}\right)^n p(\htheta\mid X) \prod_{i=1}^n \int_{\bU} \beta^{|\cY_i|}  \frac{Z(\theta)}{M}  \prod_{j=1}^{|\cY_i|}  \left\{q(y_{ji}\mid \theta) - \frac{ f(y_{ji}, \theta)}{M} \right\} \lambda(\mathrm{d}y_{ji}) \nonumber \\
    & =  \left(\frac{b_f Z(\theta)}{B_f M}\right)^n{p(\htheta\mid X)} 
    \prod_{i=1}^n \sum_{|\cY_i| = 0}^{\infty} \beta^{|\cY_i|}  \prod_{j=1}^{|\cY_i|}  \left(1 - \frac{Z(\theta)}{M} \right) \nonumber \\
& =  {p(\htheta\mid X)} \left(\frac{b_f Z(\theta)}{B_f M}\right)^n\prod_{i=1}^n \frac{1}{\tilde{\delta}_{\theta}} \qquad \tilde{\delta} = 1 - \beta(1-Z(\theta)/M) \nonumber \\
    & =   {\delta_{\theta}}{p(\htheta\mid X)} \qquad \frac{1}{\delta_{\theta}^{\frac{1}{n}}} = \frac{B_f}{b_f}\left( \frac{M}{Z(\theta)} - \beta(\frac{M}{Z(\theta)}-1)\right)
           =  \frac{B_f}{b_f}\left( \frac{M}{Z(\theta)} (1 - \beta) - \beta)\right) \nonumber \\ 
& \ge   {{\delta}}{p(\htheta\mid X)} \qquad \delta = \left\{\frac{b_f}{B_f\left( \beta + R^{-1}\right)}\right\}^n
\end{align}
Thus $k(\htheta\mid \theta)$ satisfies equation \eqref{eq:unif_erg}, with $ \delta = \left\{ 
\frac{b_f}{B_f\left( \beta + R^{-1}\right)}\right\}^n$, and $h(\htheta)$, the
posterior $p(\htheta\mid X)$.
\end{prf1}
 \label{sec:proofs}

\section{Gradient information}

For $n$ pairs $\{X_i, \cY_i\}$, 
with $N = n + \sum_{i=1}^n |\mathcal{Y}_i|$, and
$S = \sum_{i=1}^n(X_i + \sum_{j=1}^{|\mathcal{Y}_i|} Y_{ij})$, we have 
\begin{align}
 \log\left\{P(\{X_i,\cY_i\})\right\} &= \text{trace}(G^T \bkappa S) +\sum_{i=1}^n \sum_{j=1}^{|\cY_i|}\left[ \log \left\{D(\bkappa) - D(Y_{ij}, \bkappa) \right\}\right.  
                        - \left. \log\left\{D(Y_{ij}, \bkappa)\right\}\right] - N \log\left\{D(\bkappa)\right\}. \nonumber
\end{align}
Write $D(Y, \bkappa) =
 \left\{C\prod_{r=1}^p \frac{ I_{(d-r-1)/2}(\| \kappa_r N^T_r G_r \|)}{ \|\kappa_r N^T_r G_r \|^{(d-r-1)/2 }} \right\}$ as $C\tD(Y, \bkappa)
$. Since $\frac{\dif }{\dif x}\left\{\frac{I_m(x)}{x^{m}} \right\} = x^{-m}I_{m+1}(x)$, 
\begin{align}
\quad  \frac{\dif \tD(Y, \bkappa)}{\dif \kappa_j} & = N^T_jG_j \tD(Y,\bkappa) \frac{I_{(d-j+1)/2}}{I_{(d-j-1)/2}}(\kappa_j N^T_jG_j )  \quad \text{and}  
\quad  \frac{\dif \tD(\bkappa)}{\dif \kappa_j}     = \tD(\bkappa) \frac{I_{(d-j+1)/2}}{I_{(d-j-1)/2}}(\kappa_j). \nonumber
\intertext{Then, writing $L = \log\left\{P(\{X_i,\cY_i\})\right\}$, we have}
\frac{\dif L }{\dif \kappa_k} &= G_{[,k]}^T S_{[,k]} +\sum_{i=1}^n \sum_{j=1}^{|\cY_i|}\left\{\frac{\tD'(\bkappa) - \tD'(Y_{ij}, \bkappa) }{\tD(\bkappa) - \tD(Y_{ij}, \bkappa)} -
                             \frac{\tD'(Y_{ij}, \bkappa) }{\tD(Y_{ij}, \bkappa) }  \right\} - N\frac{\tD'(\bkappa)}{\tD(\bkappa)}  \nonumber\\
    &\hspace{-.2in}=  G_{[,k]}^T S_{[,k]} +\sum_{i=1}^n \sum_{j=1}^{|\cY_i|}\left\{(\frac{ \frac{I_{(d-k+1)/2}}{I_{(d-k-1)/2}}(\kappa_k)  - N^T_kG_k \frac{I_{(d-k+1)/2}}{I_{(d-k-1)/2}}(\kappa_k N^T_kG_k ) }
                             { 1 - \frac{\tD(Y_{ij}, \bkappa)}{\tD(\bkappa)}}  \right\}  
        - N  \frac{I_{(d-k+1)/2}}{I_{(d-k-1)/2}}(\kappa_k) \nonumber
\end{align}

 \label{sec:gradient}

\bibliographystyle{apalike}
\bibliography{refs}

\end{document}